\documentclass[letterpaper,USenglish,cleveref, autoref, thm-restate]{lipics-v2021}

\pdfoutput=1 
\hideLIPIcs  

\title{Efficient Caching with Reserves via Marking\footnote{An extended abstract is to appear in the Proceedings of the 50th ICALP, 2023.}}

\author{Sharat Ibrahimpur}{Department of Mathematics, London School of Economics and Political Science, UK}{s.ibrahimpur@lse.ac.uk}{}{Received funding from the following sources: NSERC grant 327620-09 and an NSERC DAS Award, European Research Council (ERC) under the European Union’s Horizon 2020 research and innovation programme (grant agreement no. ScaleOpt–757481), and Dutch Research Council NWO Vidi Grant 016.Vidi.189.087.}

\author{Manish Purohit}{Google Research, USA}{mpurohit@google.com}{}{}

\author{Zoya Svitkina}{Google Research, USA}{}{}{}

\author{Erik Vee}{Google Research, USA}{erikvee@google.com}{}{}

\author{Joshua R. Wang}{Google Research, USA}{joshuawang@google.com}{}{}

\authorrunning{S. Ibrahimpur, M. Purohit, Z. Svitkina, E. Vee, and J.\,R. Wang}

\Copyright{Sharat Ibrahimpur, Manish Purohit, Zoya Svitkina, Erik Vee, and Joshua R. Wang}


\category{} 

\relatedversion{} 

\nolinenumbers 

\EventEditors{Kousha Etessami, Uriel Feige, and Gabriele Puppis}
\EventNoEds{3}
\EventLongTitle{50th International Colloquium on Automata, Languages, and Programming (ICALP 2023)}
\EventShortTitle{ICALP 2023}
\EventAcronym{ICALP}
\EventYear{2023}
\EventDate{July 10--14, 2023}
\EventLocation{Paderborn, Germany}
\EventLogo{}
\SeriesVolume{261}
\ArticleNo{91}

\usepackage[ruled,vlined,linesnumbered]{algorithm2e}
\SetAlFnt{\small}
\usepackage{bbold}
\usepackage{float}

\allowdisplaybreaks

\newcommand{\R}{\mathbb{R}}
\newcommand{\Rp}{\R_{\geq 0}}

\newcommand{\mc}{\mathcal}
\newcommand{\A}{\mc A}

\newcommand{\I}{\mc I}
\newcommand{\T}{\mc T}
\newcommand{\mcP}{\mathcal{P}}
\newcommand{\U}{\mcP}

\newcommand{\X}{\mc X}

\newcommand{\cost}{\ensuremath{\mathrm{cost}}}
\newcommand{\OPT}{\ensuremath{\mathrm{OPT}}}

\newcommand{\full}{non-isolated}

\newcommand{\E}{\ensuremath{\mathbb{E}}}

\newcommand{\unmarked}{\ensuremath{\mathit{unmarked}}}
\newcommand{\marked}{\ensuremath{\mathit{marked}}}
\newcommand{\fettered}{isolated}

\newcommand{\tight}{tight}
\newcommand{\loose}{non-tight}

\newcommand{\ignore}[1]{}

\newcommand{\CachingWithReserves}[0]{\emph{Caching with Reserves}}
\newcommand{\Instance}[0]{\sigma}
\newcommand{\dualviolation}{5}
\newcommand{\dualcost}{\mathrm{dual}}
\newcommand{\RM}{\textrm{RM}}
\newcommand{\FM}{\textrm{FM}}

\newcommand{\bone}{\mathbb{1}}
\newcommand{\update}{\mathsf{update}}
\begin{document}

\maketitle

\begin{abstract}
Online caching is among the most fundamental and well-studied problems in the area of online algorithms. Innovative algorithmic ideas and analysis --- including potential functions and primal-dual techniques --- give insight into this still-growing area.
Here, we introduce a new analysis technique that first uses a potential function to upper bound the cost of an online algorithm and then pairs that with a new dual-fitting strategy to lower bound the cost of an offline optimal algorithm. We apply these techniques to the Caching with Reserves problem recently introduced by Ibrahimpur et al.~\cite{ibrahimpur2022caching} and give an $O(\log k)$-competitive fractional online algorithm via a marking strategy, where $k$ denotes the size of the cache. We also design a new online rounding algorithm that runs in polynomial time to obtain an $O(\log k)$-competitive randomized integral algorithm. Additionally, we provide a new, simple proof for randomized marking for the classical unweighted paging problem.
\end{abstract}

\section{Introduction} \label{sec:intro}

Caching is a critical component in many computer systems, including computer networks, distributed systems, and web applications. The idea behind caching is simple: store frequently used data items in a cache so that subsequent requests can be served directly from the cache to reduce the resources required for data retrieval. In the classical unweighted caching problem, a sequence of page requests arrives one-by-one and an algorithm is required to maintain a small set of pages to hold in the cache so that the number of requests not served from the cache is minimized.

Traditional caching algorithms, both in theory and practice, are designed to optimize the global efficiency of the system and aim to maximize the \emph{hit rate}, i.e., fraction of requests that are served from the cache. However, such a viewpoint is not particularly suitable for cache management in a multi-user or multi-processor environment. Many cloud computing services allow multiple users to share the same physical workstations and thereby share the caching system. In such multi-user environments, traditional caching policies can lead to undesirable outcomes  as some users may not be able to reap any benefits of the cache at all. Recently, Ibrahimpur et al.~\cite{ibrahimpur2022caching} introduced the \CachingWithReserves{} model that ensures certain user-level fairness guarantees while still attempting to maximize the global efficiency of the system. In this formulation, a cache of size $k$ is shared among $m$ agents and each agent $i$ is guaranteed a reserved cache size of $k_i$. An algorithm then attempts to minimize the total number of requests that are not served from the cache while guaranteeing that any time step, each agent $i$ holds at least $k_i$ pages in the cache. Unlike the classical paging problem, \CachingWithReserves{} is NP-complete even in the offline setting when the algorithm knows the entire page request sequence ahead of time and Ibrahimpur et al.~\cite{ibrahimpur2022caching} gave a 2-approximation algorithm. They also gave an $O(\log k)$-competitive online fractional algorithm for \CachingWithReserves{} via a primal-dual technique and then design a rounding scheme to obtain an $O(\log k)$-competitive online randomized algorithm. Unfortunately, the rounding scheme presented in \cite{ibrahimpur2022caching} does not run in polynomial time and the fractional primal-dual algorithm, while simple to state, also does not yield itself to easy implementation. 

Caching and its many variants have been among the most well-studied problems in theoretical computer science. It has long been a testbed for novel algorithmic and analysis techniques and it has been investigated via general techniques such as potential function analysis, primal-dual algorithms, and even learning-augmented algorithms. For the classical unweighted caching problem, a particularly simple algorithm,  \emph{randomized marking}~\cite{fiat1991competitive}, is known to yield the optimal competitive ratio (up to constant factors). At any point in time, the randomized marking algorithm partitions the set of pages in cache into \emph{marked} and \emph{unmarked} pages and upon a cache miss, it evicts an unmarked page chosen uniformly at random. Cache hits and pages brought into the cache are marked. When a cache miss occurs, but there are no more unmarked pages, a new \emph{phase} begins, and all pages in the cache become unmarked. In this paper, we build upon this algorithm, adapting it to caching with reserves.

\subsection*{Our Contributions}
\label{sec:contributions}

We study the \CachingWithReserves{} model of Ibrahimpur et al.\ \cite{ibrahimpur2022caching} in the online setting and improve upon those results. Our first main result is a simpler fractional algorithm that is a generalization of randomized marking for classical caching.

\begin{restatable}{theorem}{fractionalmainthm}
\label{thm:fraconline}
  There is an $O(\log k)$-competitive fractional \emph{marking} algorithm for online \CachingWithReserves{}. The competitive guarantee holds even when the optimal offline algorithm is allowed to hold fractional pages in the cache.
\end{restatable}

We remark that our algorithm in Theorem~\ref{thm:fraconline} and its analysis are more involved than those of the classical randomized marking algorithm. 
One complication is that due to the reserve constraints, a marking-style algorithm for the caching with reserves setting cannot evict an arbitrary unmarked page.
Another key difficulty comes from the fact that even the notion of a \emph{phase} is non-trivial to define in our setting. 
In particular, unlike in classical caching, it can happen that the cache still contains unmarked pages, but none of them can be evicted to make space for a new page, because of the reserve constraints.
Thus, we need a rule to isolate agents whose reserve constraints prevent the algorithm from having a clean end of a phase, while also ensuring that the already marked pages of such isolated agents are not erased prematurely. To this end, we introduce the notion of \emph{global} and \emph{local} phases to effectively model the state of each agent. We elaborate on this in Section~\ref{subsec:waterfilling}. 

Our analysis of the fractional marking algorithm introduces two novel components that may be of independent interest. First, we upper-bound the total cost incurred by our fractional marking algorithm using a new potential function. This potential function, introduced in Section \ref{sec:potential-function}, depends only on the decisions of the algorithm and is independent of the optimal solution. 
To the best of our knowledge, all previous potential function based analyses of (variants of) caching~\cite{bansal2010simple,bansal2010towards,bansal2022learning,ibrahimpur2022caching} define a potential function that depends on the optimal solution.
Second, we introduce a new lower bound for the cost of the optimal solution via the dual-fitting method. 
Our techniques also yield a new simple proof that the classical \emph{randomized marking}~\cite{fiat1991competitive} for unweighted paging is $O(\log k)$-competitive (see Appendix \ref{app:regular-paging}). 

We also design a new online rounding algorithm that converts a deterministic, fractional algorithm into a randomized, integral algorithm while only incurring a constant factor loss in the competitive ratio. Via a careful discretization technique (inspired by Adamaszek et al.\ \cite{adamaszek2018log}), the new rounding algorithm runs in polynomial time and only uses limited randomization.
Our fractional marking algorithm (Algorithm \ref{alg:water-filling}) maintains that at any point in time, a particular page $p$ is either completely in the cache or at least $1/k$ fraction of the page has been evicted. We exploit this key property to show that the fractional solution at any time $t$ can be discretized so that the fraction of any page that is evicted is an integral multiple of $1/k^3$. This discretization allows us to maintain a distribution over feasible integer cache states with bounded support.

\begin{restatable}{theorem}{integralmainthm}
\label{thm:randomized}
There is a polynomial-time $O(\log k)$-competitive randomized integral algorithm for online caching with reserves.
\end{restatable}

\subsection*{Other Related Work}
\label{sec:related}

The unweighted caching (also known as paging) problem has been widely studied and its optimal competitive ratio is well-understood even up to constant factors. Tight algorithms \cite{achlioptas2000competitive,mcgeoch1991strongly} are known that yield a competitive ratio of exactly $H_k$, where $H_k$ is the $k$th harmonic number. Recently, Agrawal et al.~\cite{agrawal2021tight} consider the \emph{parallel paging} model where a common cache is shared among $p$ agents -- each agent is presented with a request sequence of pages and the algorithm must decide how to partition the cache among agents at any time. It allows the $p$ processors to make progress simultaneously, i.e., incur cache hits and misses concurrently. Multi-agent paging has also been extensively studied in the systems community~\cite{chang2007cooperative,stone1989optimal,suh2004dynamic} often in the context of caching in multi-core systems. Closely related to the \CachingWithReserves{} setting, motivated by fairness constraints in multi-agent settings, a number of recent systems ~\cite{wu2022NyxCache,kunjir2017robus,pu2016fairride,yu2019lacs} aim to provide \emph{isolation guarantees} to each user, i.e., guarantee that the cache hit rate for each user is at least as much as what it would be if each user is allocated its own isolated cache. Also motivated by fairness constraints, Chiplunkar et al.~\cite{chiplunkar2023online} consider the Min-Max paging problem where the goal is to minimize the maximum number of page faults incurred by any agent.

\section{Preliminaries and Notation}
\label{sec:prelims}

Formally, an instance of the \CachingWithReserves{} problem consists of the following. We are given a number of agents $m$ and a total (integer) cache capacity $k$. Let $[m]$ denote the set $\{1, \ldots, m\}$. Each agent $i \in [m]$ owns a set of pages $\U(i)$ (referred to as $i$-pages) and has a \emph{reserved cache size} $k_i \ge 0$. Pages have a unique owner, i.e. $\U(i) \cap \U(j) = \varnothing$ for all $i \ne j$, and we use $\U \triangleq \cup_{i \in [m]} \U(i)$ to refer to the universe of all pages. For any page $p \in \U$, let $ag(p)$ be the unique agent that owns $p$. We assume without loss of generality that at least one unit of cache is not reserved: $\sum_{i \in [m]} k_i < k$.\footnote{If all of cache is reserved, the problem decomposes over agents into the standard caching task.} At each timestep $t$, a page $p_t \in \U$ is requested. We can wrap all these into an instance tuple: $\Instance = \left( m, k, \{\U(i)\}, \{k_i\}, \{p_t\} \right)$.

An \emph{integral} algorithm for the \emph{Caching with Reserves} problem maintains a set of $k$ pages in the cache such that for each agent $i$, the cache always contains at least $k_i$ pages from $\U(i)$. At time $t$, the page request $p_t$ is revealed to the algorithm. If this page is not currently in the cache, then the algorithm is said to incur a \emph{cache miss} and it must fetch $p_t$ into the cache by possibly evicting another page $q_t$. For any (integral) algorithm $\A$ we write its total cache misses on instance $\Instance$ as $\cost_\A(\Instance)$.

A \emph{fractional} algorithm for the \emph{Caching with Reserves} problem maintains a fraction $x_p \in [0, 1]$ for how much each page $p \in \U$ is in the cache such that the total size of pages in the cache is at most $k$, i.e, $\sum_{p \in \U} x_p \le k$, and the total size of $i$-pages is at least $k_i$, i.e., $\sum_{p \in \U(i)} x_p \ge k_i$. At time $t$, the page request $p_t$ is revealed to the algorithm, which incurs a fractional cache miss of size $1 - x_{p_t}$. The algorithm must then  fully fetch $p_t$ into cache ($x_{p_t} \leftarrow 1$) by possibly evicting other pages. For any (fractional) algorithm $\A$ we again write its total size of cache misses on instance $\Instance$ as $\cost_\A(\Instance)$.

Let $\cost_{OPT}(\sigma)$ be the cost of the optimal offline algorithm on instance $\sigma$.

\begin{definition}[Competitive Ratio]
An online algorithm $\A$ for \CachingWithReserves{} is said to be $c$-competitive, if for any instance $\sigma$, $\E[\cost_{\A}(\sigma)] \leq c \cdot \cost_{OPT}(\sigma) + b$, where $b$ is a constant independent of the number of page requests in $\sigma$. The expectation is taken over all the random choices made by the algorithm (if any).
\end{definition}
\section{Fractional \boldmath \texorpdfstring{$O(\log k)$}{O(log k)}-Competitive Algorithm for Caching with Reserves}

For any time $t$ and page $p \in \U$, the algorithm maintains a variable $y_p^t \in [0,1]$ representing the portion of page $p$ that is outside the cache. Then $x_p^t \triangleq 1 - y_p^t$ represents the portion of $p$ that is in cache. 
Algorithm \ref{alg:water-filling} ensures feasibility at all times $t$: the total of all $y$ values is exactly the complementary cache size $|\U| - k$, i.e., $\sum_{p \in \U} y_p^t = |\U| - k$; and the total $y$ value for pages of any agent $i$ is within its respective complementary reserve size, $\sum_{p \in \U(i)} y_p^t \leq |\U(i)| - k_i$. When a request for page $p_t$ arrives at time $t$, the algorithm fully fetches $p_t$ into the cache by paying a fetch-cost of $y_{p_t}^t$ while simultaneously evicting a total of $y_{p_t}^t$ amount of other suitably chosen pages. 

\subsection{Fractional Algorithm} \label{subsec:waterfilling}

The complete algorithm (referred to as Algorithm $\A$ in the proofs) is presented in Algorithm~\ref{alg:water-filling}. We present a high-level discussion here. At any time $t$, we say that an agent $i$ is \emph{tight} if $\sum_{p \in \U(i)} x_{p}^t = k_i$, i.e., the algorithm is not allowed to further evict any pages (even fractionally) of agent $i$. Conversely, an agent $i$ is \emph{non-tight} if $\sum_{p \in \U(i)} x_{p}^t > k_i$.

The algorithm is a fractional marking algorithm and runs in phases where each phase corresponds to a maximal sequence of page requests that can be served while maintaining feasibility and ensuring that no ``marked'' pages are evicted. 
Within each phase, the currently requested page $p_t$ is fully fetched into cache by continuously evicting an infinitesimal amount of an ``available'' (described below) unmarked page $q$ with the smallest $y_q$ value; if there are multiple choices of $q$, then all of them are simultaneously evicted at the same rate. 
Page $p_t$ gets marked after it has been served and this mark may only be erased at the end of a phase. 

At the end of a phase, an agent $i$ is designated as \emph{isolated} if strictly fewer than $k_i$ $i$-pages are marked in the cache at this time point. This designation changes to non-isolated as soon as $k_i$ $i$-pages get marked at some point in the future. An isolated agent essentially runs a separate instance of caching on its own pages and in its own reserved space. At the end of a phase, the marks of pages owned by non-isolated agents (i.e., agents with at least $k_i$ marked $i$-pages) are erased. 

It remains to describe when a page $q$ is considered available for eviction. Clearly, $y^t_q < 1$ must hold, since otherwise page $q$ is already fully outside the cache. Moreover,
$ag(q)$ must be \emph{non-tight}, i.e., evicting page $q$ must not violate the reserve constraint of the agent that owns it.
The last condition for $q$ to be considered available for eviction depends on whether the agent $i_t := ag(p_t)$ is \emph{isolated} or not: (i) if agent $i_t$ is isolated, then $ag(q) = i_t$ should hold, i.e., only unmarked $i_t$-pages are available for eviction; and (ii) if agent $i_t$ is not isolated, then $ag(q)$ should also be non-isolated. We recall again that among all available pages for eviction, pages with the smallest $y^t_q$ value are evicted first.

\begin{algorithm}[h]
\caption{Fractional Marking Algorithm for \CachingWithReserves{}}\label{alg:water-filling}
\DontPrintSemicolon
\texttt{/* Initialization */}\;
$r_0 \leftarrow 1$ \tcc*{global phase counter}
$r_i \leftarrow 1,\ \forall i \in [m]$ \tcc*{local phase counters}
Let $P(i,0) \subset \U(i)$ be set of $i$-pages in the initial cache (assume $|P(i,0)| \geq k_i$)$\ \forall i \in [m]$ \;
All agents $i \in [m]$ are \emph{\full} and all pages $p$ in the cache are \unmarked\;

\For{each page request $p_t$ of agent $i_t$}{
\uIf{$y_{p_t} = 0$, i.e., $x_{p_t} = 1$,}{\emph{Mark} page $p_t$ and serve the request.}
\uElseIf{agent $i_t$ is \fettered,}{
\texttt{/* Continuously fetch page $p_t$ while uniformly evicting all unmarked $i_t$-pages. */}\;
Set $y_{p_t} \leftarrow 0$, \emph{mark} page $p_t$ and serve the request. \;
Increase $y_q$ at the same rate for all unmarked $i_t$-pages in the cache until the cache becomes feasible, i.e. $\sum_{p \in \U} y_p \geq |\U| - k$ holds. \label{algline:evict-isolated}\;
\If{agent $i_t$ now has $k_i$ marked pages} {
    Designate $i_t$ as \full\;
}
}
\uElseIf{$\exists$ page $q$ owned by some \loose\  agent\footnotemark and satisfying $y_q < 1$,}{
\texttt{/* Continuously fetch page $p_t$ while uniformly evicting all unmarked pages (belonging to any \loose\ agent) with the least $y$-value.  */}\;
Set $y_{p_t} \leftarrow 0$, \emph{mark} page $p_t$ and serve the request. \;
Increase $y_q$ at the same rate for all unmarked pages of \loose\ agents with the smallest $y$ values until the cache becomes feasible, i.e. $\sum_{p \in \U} y_p \geq |\U| - k$ holds. \label{algline:evict-non-isolated}\;
}
\Else{
\texttt{/* End of phase */}\;
\For{each agent $i \in [m]$}{
\uIf{$i$ has strictly fewer than $k_i$ marked pages,}{Designate $i$ as \fettered.\;}
\Else{
\texttt{/* $i$ is \full\ and undergoes a phase reset */}\;
Set $P(i,r_i) \leftarrow $ collection of all (integral and marked) $i$-pages in cache.\;
Set $r_i \leftarrow r_i + 1$\;
All marked $i$-pages are now \unmarked\;
}
}
Set $r_0 \leftarrow r_0 + 1$\;
Re-process the current page request $p_t$ in the new phase
\label{algline:phase-end}\;
}
}
\end{algorithm}

\footnotetext{Agent $i_t$ is considered non-tight here because fetching $p_t$ while evicting other $q \in \U(i_t) \setminus p_t$ does not violate reserve  feasibility.}


\subparagraph*{Notation}
Let $\I(t) \subsetneq [m]$ denote the set of \fettered\ agents at time $t$. For a global phase $r_0$, we use $\I(r_0)$ to denote the set of \fettered\ agents at the end of phase $r_0$. Let $\T(t)$ denote the set of \emph{tight} agents at time $t$. At any time $t$, let $r_i^t$ denote the value of the local phase counter $r_i$ for agent $i$, and let $R_i = r_i^T$ be the total number of local phases for agent $i$. By definition, for any agent $i \in [m]$, $P(i,r_i-1)$ and $P(i,r_i)$ denote the set of $i$-pages in the cache (integrally) at the beginning of the $r_i$th local phase and the end of the $r_i$th local phase for agent $i$, respectively. For any agent $i \in [m]$, let $M(i,t) \subseteq \U(i)$ denote the set of marked $i$-pages in the cache and $U(i,t) = P(i,r_i^t - 1) \setminus M(i,t)$ denote the set of unmarked $i$-pages.

We emphasize that the notion of \unmarked\ pages will only be relevant while referring to pages in $P(i,r^t_i-1)$ for some $i,t$; in particular, every $i$-page $q \in \U(i) \setminus (P(i,r^t_i-1) \cup M(i,t))$ is not \marked, but we do not refer to it as \unmarked. Analogous to the notion of clean and stale pages used by the randomized marking algorithm~\cite{fiat1991competitive}, we define \emph{clean, pseudo-clean} and \emph{stale} pages as follows. Fix an agent $i$ and let $r_i$ be its local phase counter at time $t$. Any $i$-page $q \in P(i,r_i-1)$ is considered \emph{stale}. The currently requested page $p_t$ is said to be \emph{clean} if $p_t \notin P(i,r_i-1)$. 
Next, we say that the currently requested page $p_t$ is \emph{pseudo-clean} if $p_t \in P(i,r_i-1)$ and $y^t_{p_t} = 1$ holds right before Algorithm~$\A$ starts to fetch $p_t$ into the cache.  Lemmas \ref{lem:tight} and \ref{lem:frontier} show that a pseudo-clean page necessarily belongs to an agent who was \fettered\ at the start of (global) phase $r_0$ but is \full\ at time $t$. To simplify notation, we drop the superscript $t$ from all notation whenever the time index is clear from the context. 

The following lemma compiles a list of key invariants that are maintained throughout the execution of the algorithm that follow directly from an examination of Algorithm \ref{alg:water-filling}. 

\begin{lemma}\label{lem:invariants}
Algorithm \ref{alg:water-filling} maintains the following invariants.
\begin{enumerate}[(i)]
    \item When a new phase begins, all marked pages belong to \fettered\ agents.
    
    \item \label{inv:tight} At any time $t$, all \fettered\ agents are tight.
    
    \item \label{inv:samevalue} At any time $t$ and for any agent $i$, all unmarked pages of agent $i$ have the same $y$ value.
    
    \item Any page belonging to an \fettered\ agent is (fractionally) evicted only in those timesteps when a different page of the same agent has been requested.
\end{enumerate}
\end{lemma}

The following lemmas show that the algorithm is well-defined and that the operations in Lines~\ref{algline:evict-isolated} and \ref{algline:evict-non-isolated} of Algorithm \ref{alg:water-filling} are always feasible. 

\begin{lemma} \label{lem:fettered-mass}
If the requested page $p_t$ has $x_{p_t}^t \in [0, 1)$ and agent $i_t$ is \fettered, then $p_t$ can be fetched fully by evicting unmarked pages of agent $i_t$.
\end{lemma}

\begin{proof}
As agent $i_t$ is isolated when page $p_t$ is requested, $\sum_{p \in \U(i)} x_p^t = k_i$ (by invariant (\ref{inv:tight}) in Lemma \ref{lem:invariants}) and  $i_t$ has fewer than $k_i$ marked pages in cache. Hence $\sum_{p \in U(i,t)} x_p^t \geq 1$ and $\sum_{p \in U(i,t) \setminus \{p_t\}} x_p^t \geq 1 - x_{p_t}^t$. 
\end{proof}

\begin{lemma} \label{lem:full-mass}
If the requested page $p_t$ has $0 < x^t_{p_t} < 1$ and its owner $i_t$ is \full, then there is always enough fractional mass of pages belonging to \loose\ agents that can be evicted to fully fetch page $p_t$. In particular, line~\ref{algline:evict-non-isolated} of Algorithm \ref{alg:water-filling} is well-defined.
\end{lemma}

\begin{proof}
Suppose page $p_t$ is fetched in a continuous manner. To show that page $p_t$ can be fetched fully, it suffices to show that at any instantaneous time $t$ when $x_{p_t}^t < 1$, there always exists an unmarked page $q$ belonging to a non-tight agent $i$ such that $x_q^t > 0$, i.e. page $q$ can be evicted. Observe that $k = \sum_{q \in \U} x_{q}^t = \sum_{i \in \T(t)} k_i + \sum_{i \notin \T(t)} \sum_{q \in \U(i)}x_{q}^t$. Due to integrality of $k$ and $\{k_i\}_{i \in [m]}$, we must have $\sum_{i \notin \T(t)} \sum_{q \in \U(i)}x_{q}^t$ is an integer. Since any marked page $q$ always has $x_q^t = 1$, $\mu := \sum_{i \notin \T(t)} \sum_{q \in U(i, t)}x_{q}^t$ is also an integer. Further, since $i_t \notin \T(t)$ and $x_{p_t}^t > 0$, we must have $\mu \geq 1$ and hence there exists a page $q$ belonging to some \loose\ agent $i$ with $x_q^t > 0$ as desired.
\end{proof}

Since we always evict an available page with the least $y$ value, at any time step $t$, all available pages (i.e., unmarked pages $q$ belonging to \loose\ agents and satisfying $y_q < 1$) have the same $y$ value at all times. We denote this common $y$-value by $h^*$ and refer to the corresponding set of evictable pages (with $y$-value $h^*$) as the \emph{frontier}. The following two key structural lemmas formalize this property.

\begin{lemma} \label{lem:tight} 
At any time $t$, let $i$ be an agent that was isolated at the beginning of the current phase, and let $q$ be one of its unmarked pages. Then $y^t_q < 1$ if and only if $i$ is still \fettered\ at time $t$.
\end{lemma}

\begin{proof}
For the \emph{if} direction, suppose that $i$ is still \fettered. By invariant (\ref{inv:tight}), it is also tight. By invariant (\ref{inv:samevalue}), all its unmarked pages have the same $x$-value and in total they occupy $k_i - |M(i,t)| > 0$ units of cache space. Thus, $x_q^t > 0$ and $y^t_q < 1$. 

For the \emph{only if} direction, suppose that $i$ is no longer isolated. Just before the $k_i$th $i$-page to be marked was requested, $(k_i-1)$ $i$-pages were marked. Since $i$ was tight, the total $x$ value of all its unmarked pages must have been 1. Then the algorithm replaced all of them with the $k_i$th marked page, and the $x$-value of all remaining unmarked pages became $0$. Thus, the property holds for a newly \full\ agent $i$. This property continues to hold for the rest of the phase since $y_q$ never decreases for an unmarked page.
\end{proof}

\begin{lemma} \label{lem:frontier}
At any time $t$, there is a value $h^*_t \in [0,1]$ such that:
for any agent $i$ that was non-isolated at the beginning of the current phase and 
any unmarked $i$-page $q$,  $y_q \leq h^*_t$ holds and $y_q = h^*_t$ holds whenever $i$ is \loose.
\end{lemma}

\begin{proof}
We prove the lemma by induction. Clearly, the lemma holds at the start of the phase: all unmarked pages belonging to \full\ agents have $y$-value $0$. Now consider a time $t$ during phase $r$ such that the lemma holds for all timepoints before $t$ in this phase. We may also assume that $y_{p_t} > 0$, since otherwise none of the variables are modified in this timestep.

By induction hypothesis, any unmarked $i$-page $q$ satisfies $y_q \leq h^*_{t-1}$, and this inequality is tight whenever $i$ is \loose. If agent $i_t$ is \loose, then its unmarked pages are already part of the frontier. Otherwise, Algorithm~$\A$ fetches $p_t$ fully into the cache by increasing the $y$-value of other unmarked $i_t$-pages until one of the following happens: (a) $p_t$ is fully fetched. In this case, $i_t$ continues to remain \tight; or (b) The $y$-value of unmarked $i_t$-pages becomes equal to the frontier's $y$-value, $h^*_{t-1}$. In the latter case, unmarked $i_t$-pages become part of the frontier and the $y$-value of the frontier is uniformly increased until $p_t$ gets fully fetched into the cache. If some agent $i'$ becomes \tight\ before the fetch operation is completed, then its unmarked pages get excluded from the frontier and the corresponding $y$-values remain unchanged for the rest of this timestep. In all cases, the lemma continues to hold since the $y$-value of the frontier is never decreased and only \tight\ agents get dropped from the frontier.
\end{proof}

\begin{remark}
Within any phase, $h^*_t$ is non-decreasing over time and takes values $0$ and $1$ at the endpoints. This follows from the fact that $\A$ never decreases $y_q$ for an unmarked page $q \neq p_t$.
\end{remark}

The following lemma shows that any page that is not completely in Algorithm~$\A$'s cache must be evicted to at least a $1/k$ portion. This property will be useful to us in Sections~\ref{sec:potential-function}~and~\ref{sec:rounding}. 

\begin{lemma} \label{lem:ylarge}
At the end of any time step $t$, for any page $p \in \U$, we have $y_p^t = 0$ or $y_p^t \geq 1/k$.
\end{lemma}

\begin{proof}
First, note that for all marked pages, we have $y_p^t = 1 - x_p^t = 0$. Let $i \in \T(t)$ be any \tight\ agent. Then we have $k_i = \sum_{p \in \U(i)} x_p^t = |M(i,t)| + \sum_{p \in U(i,t)} x_p^t$. By Lemma \ref{lem:invariants} (part \ref{inv:samevalue}), all unmarked pages of agent $i$ have the same $y$ value: $y_p^t = 1 - x_p^t = h_i$ (say). Since $k_i$ is integral, we have either $h_i = 0$ or $h_i$. Rearranging, we have $|U(i,t)| h_i = |M(i,t)| + |U(i,t)| - k_i$. Since all terms on the RHS are integral, we have either $h_i = 0$ or $h_i \geq 1/|U(i,t)| \geq 1/k$.

By Lemma \ref{lem:frontier}, all unmarked pages $p$ belonging to \loose\ agents satisfy $y_p^t = h_t^*$. Let $U(t)$ be the set of all unmarked pages belonging to all \loose\ agents that were also \full\ at the beginning of the phase. Recall that by definition, we have $|U(t)| \leq k$ since all pages in $U(t)$ must have been fully in the cache at the beginning of the current phase. Since we have $k = \sum_{i \in \T(t)} k_i + \sum_{i \notin \T(t)} \sum_{p \in \U(i)} x_p^t$, once again by integrality of $k$ and $\{k_i\}$, we must have that $\sum_{p \in U(t)} y_p^t = \sum_{p \in U(t)} h_t^*$ is an integer. Hence, either $h_t^* = 0$ or $h_t^* \geq 1/|U(t)| \geq 1/k$.
\end{proof}

\subsection{Analysis Overview}
At any time $t$, we consider the set of $y$ values of pages in $\bigcup_{i \in [m]} P(i,r_i - 1)$ as the \emph{state} of the system. We define a non-negative potential function $\Psi$ that is purely a function of this state. For any page request $p_t$, we attempt to bound the algorithm's cost by an increase in the potential function, thereby bounding the total cost incurred by the algorithm by the final value of the potential function. There are two difficulties with this approach: (i) when a phase ends, the potential function abruptly drops since all the unmarked pages that were fully evicted no longer contribute to the state, and (ii) when the agent $i_t$ was \fettered\ at the beginning of the phase but is now \full, the change in potential is not sufficient to cover the fetch cost. In both these situations we charge the cost incurred by the online algorithm to a new quantity that is a function of the sets $\{P(i,r_i)\}$. To complete the analysis, we show that this quantity is upper-bounded by the cost of the optimal solution.

\subsection{Potential Function Analysis}
\label{sec:potential-function}

Consider the function $\phi : [0,1] \to \Rp$ defined as:
\begin{equation} \label{eq:phipot}
\phi(h) \triangleq 2h \cdot \ln(1+kh)
\end{equation}
As $h$ goes from $0$ to $1$, $\phi(h)$ increases from $0$ to $2\ln(1+k)$. 

The potential at any time $t$ is defined as follows:
\begin{equation}
    \Psi(t) \triangleq \sum_{i=1}^{m} \sum_{p \in U(i,t)} \phi(y_p^t)
\end{equation}
Note that only unmarked pages at any time $t$ contribute to the potential. So when page $p_t$ is fetched at time $t$ and marked, it stops contributing to the potential. But since $\phi$ is monotone, the newly evicted pages increase their contribution to the potential. We remark that the potential is purely a function of the state of the system as defined by the $y$ values of unmarked pages in the cache and is thus always bounded by a quantity independent of the length of the page request sequence.


\begin{lemma} \label{lem:largeslope}
For any $h \geq 1/k$, we have $\phi(h) \geq h$ and $\phi'(h) \geq 1 + 2 \ln(1+kh)$. 
\end{lemma}

\begin{proof}
The first conclusion follows from the logarithmic inequality $\ln(1+x) \geq x/(1+x)$ which holds for any nonnegative $x$: we have $\phi(h) = 2h \ln(1+kh) \geq h \cdot 2kh/(1+kh) \geq h$ whenever $kh \geq 1$.
Next, $\phi'(h) = \frac{d\phi}{dh} = 2(1 - 1/(1+kh) + \ln(1+kh))$. So, for any $h \geq 1/k$ we have $1/2 \geq 1/(1+kh)$, which gives the other conclusion. 
\end{proof}

The rest of this section is devoted to proving the following theorem where we bound the total cost incurred by the algorithm in terms of the sets $\{P(i,r_i)\}$ and the number of requests to pseudo-clean pages.

\begin{theorem} \label{thm:potential-function}
The following bound holds on the cost incurred by $\A$ to process the first $T$ page requests:
\[
\cost_{\A}(\sigma) \leq 2\ln(1+k) \cdot \left(mk + \sum_{t=1}^T \bone_{{p_t} \text{ is pseudo-clean}} + \sum_{i \in [m]} \sum_{r_i=1}^{R_i} |P(i,r_i-1) \setminus P(i,r_i)|\right).
\]
\end{theorem}

Recall that the algorithm incurs a cost of $y_{p_t}^t$ to fetch page $p_t$ at time $t$. So the total cost incurred by the algorithm is simply $\cost_{\A}(\sigma) = \sum_{t} y_{p_t}^t$. We first bound this cost for time steps when the requested page $p_t$ is at least partially in the cache, i.e., $y_{p_t}^t < 1$. Recall by Lemmas \ref{lem:fettered-mass} and \ref{lem:full-mass}, the algorithm does not undergo a phase transition in this time step. 

\begin{lemma} \label{lem:paypartialfetch}
Consider any time step $t$ such that $y_{p_t}^t < 1$ for the currently requested (unmarked) page $p_t$. Let $\Delta \Psi(t)$ denote the change in the potential function during time step $t$. Then  
$y_{p_t}^t \leq \Delta \Psi(t)$.
\end{lemma}

\begin{proof}
We assume that $y_{p_t}^t \geq \frac{1}{k}$, since otherwise by Lemma \ref{lem:ylarge}, we must have $y_{p_t}^t = 0$ and the lemma follows trivially. Since $y_{p_t}^t < 1$, by Lemma \ref{lem:frontier}, either agent $i_t$ is tight or we have $y_q^t = y_{p_t}^t$ for every  unmarked page $q$ owned by any \loose\ agent $i$ that was \full\ at the start of this phase. 
In either case, the pages that get evicted to make space for $p_t$ have their initial $y$ values at least $y_{p_t}^t \geq 1/k$. The potential function $\Psi$ changes in this step due to two factors: (i) $\Psi$ drops as page $p_t$ stops contributing to the potential as soon as it gets marked; and (ii) $\Psi$ increases as the $y$-value of (fractionally) evicted pages increases in this step.

Let $h \triangleq y^t_{p_t}$. At the beginning to time $t$, page $p_t$ contributed exactly $\phi(h) = 2 h \ln(1 + kh)$ to the potential; This contribution is lost as soon as $p_t$ gets marked. To prove the lemma, it suffices to show that the rate of increase in the potential function (without including $p_t$'s contribution) is at least $1+2\ln(1+kh)$ throughout the eviction of an $h$ amount of unmarked pages belonging to \loose\ agents: the $1$ term in total pays for the fetch-cost of $h$ and the $2\ln(1+kh)$ term in total pays for the $2 h \ln(1 + kh)$ loss in potential. 
This directly follows from Lemma~\ref{lem:largeslope} from the fact that the $y$-values of pages that are fractionally evicted in this timestep were already at least $h \geq 1/k$. Here, we also use the monotonicity of the function $h' \mapsto \ln(1+kh')$.
\end{proof}

We still need to bound the cost incurred by the algorithm when the incoming request is to a page that is fully outside the cache. Note that the algorithm incurs exactly unit cost for all such time steps. The following lemma shows that the total cost incurred by the algorithm can be bounded by the drop in potential function at the end of a phase and by a term that depends only on the change in the potential function while processing a request to a page fully outside the cache.

\begin{lemma} \label{lem:cost-break-up}
For any global phase $r_0$, let $\Delta \Psi(r_0)$ denote the change in the potential function at the end of phase $r_0$ (line \ref{algline:phase-end} in Algorithm \ref{alg:water-filling}). Let $R_0$ denote the total number of global phases and $T$ denote the time at the end of phase $R_0$. Then we have the following upper bound on the cost incurred by $\A$ for processing the first $T$ page requests: 
\[
\cost_\A(\sigma) \leq 2mk \ln(1+k) + \sum_{t \in [T] : y_{p_t}^t = 1} (1 - \Delta \Psi(t)) - \sum_{r_0=1}^{R_0} \Delta \Psi(r_0)
\]
\end{lemma} 

\begin{proof}
We have: 
\begin{align*}
\cost_\A(\sigma) &= \sum_{t \in [T]} y_{p_t}^t = \sum_{t \in [T] : y_{p_t}^t < 1} y_{p_t}^t + |\{t \in [T] : y_{p_t}^t = 1\}| \\
& \leq \sum_{t : y_{p_t}^t < 1} \Delta \Psi(t) + |\{t : y_{p_t}^t = 1\}| \tag{Using Lemma~\ref{lem:paypartialfetch}} \\
& = \Psi(T) - \Psi(0) - \sum_{t : y_{p_t}^t = 1} \Delta\Psi(t) - \sum_{r_0=1}^{R_0} \Delta \Psi(r_0) + |\{t : y_{p_t}^t = 1\}|.
\end{align*}
The lemma follows since $\Psi(T) \leq 2mk \ln(1+k)$ and $\Psi(0) = 0$. The bound on $\Psi(T)$ is because we have $m$ agents each with $|P(i,r_i-1)| \leq k$, and $\phi(1) = 2\ln(1+k)$.
\end{proof}


So, it is enough to bound the total cost and drop in potential for time steps when the requested page is fully outside the cache and also to bound the drop in potential when the phase changes.

\begin{proof}[Proof of Theorem~\ref{thm:potential-function}]
Consider any time step $t$ such that the currently requested page $p_t$ is fully outside the cache, i.e. $y_{p_t}^t = 1$. We differentiate such requests into two cases depending on whether the page $p_t$ is in the set $P(i_t,r_{i_t} - 1)$ at the time or not. In other words, we do a case analysis on $p_t$ being clean or pseudo-clean. (Recall that only unmarked pages in $P(i_t,r_{i_t} - 1)$ contribute to the potential).

\medskip

\emph{Case 1: $p_t \notin P(i_t,r_{i_t} - 1)$, i.e, $p_t$ is clean.} Since page $p_t \notin U(i,t)$, it does not contribute to the potential before (or after) the request has been served. Consider any page $q$ that is evicted (fractionally) by the algorithm in this step. By Lemma \ref{lem:invariants}, before the eviction, we have $y_q = 0$ or $y_q \geq 1/k$. In either case, by Lemma \ref{lem:largeslope}, we have $\Delta  \phi(y_q) \geq \Delta y_q$ where $\Delta y_q$ denotes the change in $y$-value of page $q$ in this step. Since we have $\sum_q \Delta y_q = y_{p_t}^t = 1$, we have $\Delta \Psi(t) = \sum_{q} \Delta \phi(y_q) \geq 1$.

\smallskip

\emph{Case 2: $p_t \in P(i_t,r_{i_t}- 1)$, i.e., $p_t$ is pseudo-clean.} By the same reasoning as above, we have $\sum_{q \neq p_t} \Delta \phi(y_q) \geq 1$. However, in this case, page $p_t$ also contributed exactly $2\ln(1+k)$ to the potential at the beginning of the time step. So we have
$\Delta \Psi(t) = \sum_{q \neq p_t} \Delta \phi(y_q) - 2 \ln(1+k) \geq 1 - 2 \ln(1+k)$.

Combining the two cases we get:
\begin{equation}
    \sum_{t : y_{p_t}^t = 1}( 1 - \Delta\Psi(t)) \leq 2 \ln(1+k) \cdot |\{t : y_{p_t} = 1 \text{ and } p_t \in P(i_t,r_{i_t} - 1)\}|
\end{equation}

Consider the end of some phase $r_0$ and let $i$ be a \full\ agent. Let $r_i$ denote the current local phase of agent $i$ that must also end along with the global phase $r_0$. Consider any unmarked page $q$ in $U(i,t)$. As the phase $r_0$ is ending, page $q$ must be fully evicted and thus contributes $\phi(1)$ to the potential. Once phase $r_0$ ends and phase $r_0+1$ begins, page $q$ no longer contributes to the potential. Note that the set of such unmarked pages is exactly $P(i,r_i-1) \setminus P(i,r_i)$. Hence, the change in potential at the end of (global) phase $r_0$ is given by:  
\[
\Delta \Psi(r_0) = -2\ln(1+k) \cdot \sum_{i \notin \I(r_0)} |P(i,r_i-1) \setminus P(i,r_i)|
\]
Since an agent only changes its local phase when it is \full\ at the end of a global phase, we have:
\begin{equation}
    \sum_{r_0=1}^{R_0} \Delta \Psi(r_0) = -2\ln(1+k) \cdot \sum_{i \in [m]} \sum_{r_i=1}^{R_i} |P(i,r_i-1) \setminus P(i,r_i)|.
\end{equation}
The theorem now follows from Lemma~\ref{lem:cost-break-up}.
\end{proof}

\subsection{\boldmath A Lower Bound on \texorpdfstring{$\OPT$}{OPT} through Dual Fitting}
\label{sec:dual-fitting}

In this section, we give a novel LP-based lower bound on the cost of any  offline algorithm for caching with reserves via dual-fitting. This lower bound analysis is new even for the classical unweighted paging setting. 
Crucially, the lower bound derived here perfectly matches the two terms used to bound the cost of the fractional algorithm $\A$ in Theorem~\ref{thm:potential-function},
thereby completing the proof of our main result (Theorem~\ref{thm:fraconline}).

We now describe the linear relaxation of the caching with reserves problem and its dual program. The following notation will be useful. For any page $q \in \U$, let $t_{q,1} < t_{q,2} < ...$ denote the time steps when $q$ is requested in the online sequence. For an integer $a \geq 0$, define $I(q,a) = \{t_{q,a} + 1, \ldots, t_{q,a+1}-1\}$ to be the time interval between the $a$th and $(a+1)$th requests for $q$. We define $t_{q,0} \triangleq 0$ for all pages. Let $a(q,t)$ denote the number of requests to page $q$ that have been seen until time $t$ (inclusive). Hence, by definition, for any time $t$ and page $q \in \U \setminus \{p_t\}$, we have $t \in I(q, a(q,t))$. The primal LP has variables $y(q,a) \in [0,1]$ which denote the portion of page $q$ that is evicted between its $a$th and $(a+1)$th requests, i.e., $1-y(q,a)$ portion of $q$ is held in the cache during the time-interval $I(q,a)$. 
For convenience, we define $n \triangleq |\U|$ and $n_i \triangleq |\U(i)|$ for any $i \in [m]$. The first and second set of primal constraints encode the cache size constraint and the agent-level reserve constraints for all times. The dual LP has variables $\alpha(t)$ and $\beta(t,i)$ corresponding to these primal constraints. We also have dual variables $\gamma(q,a)$ corresponding to the primal constraint encoding $y(q,a) \leq 1$. Besides nonnegativity, the dual has a single constraint for each interval $I(q,a)$. The primal and dual LPs are stated below. We emphasize that we use these linear programs purely for analysis and the algorithm itself does not need to solve any linear program.

\medskip

\noindent
\hspace{-1em}
\begin{minipage}[t]{0.5\textwidth}
\footnotesize
\begin{center}
    \textbf{Primal LP}
\end{center}
\begin{align}
    & \min \quad \sum_{q \in \U} \sum_{\substack{a \geq 1}} y(q, a) & \notag \\
    & \text{subject to:} \notag \\
    & \quad \sum_{q \in \U, q \neq p_t} y(q, a(q,t)) \geq n - k & \forall t \label{eq:primalcachesize} \\
    & \sum_{q \in \U(i), q \neq p_t} y(q, a(q,t)) \leq n_i - k_i & \forall t, \forall i \label{eq:primalreserve} \\
    & \hspace{7em} y(q,a) \leq 1 & \forall q, \forall a \label{eq:primalvarbound} \\
    & \hspace{9em}\  y \geq 0
\end{align}
\end{minipage}
\quad\vline\hspace{-1em}
\begin{minipage}[t]{0.5\textwidth}
\footnotesize
\begin{center}
    \textbf{Dual LP}
\end{center}
\begin{align}
&\max \sum_t (n-k) \alpha(t) - \sum_{t,i} (n_i - k_i) \beta(t,i) \notag\\
&\hspace{9em} - \sum_{q,a} \gamma(q,a) \notag \\
&\text{subject to:} \notag \\
&\sum_{t \in I(q,a)} \bigl(\alpha(t) - \beta(t, ag(q)) \bigr) - \gamma(q,a) \notag \\ & \hspace{100pt} \quad \leq 1 \, \, \forall q, \forall a \label{eq:dualconstraint} \\
& \hspace{9em}\ \alpha,\beta, \gamma \geq 0 
\end{align}
\end{minipage}

\medskip
\medskip
\medskip

Consider time $T$ that marks the end of a global phase $R_0$ for some integer $R_0$. 
Let $\OPT = \cost_{\OPT}(\sigma)$ denote the total cost incurred by an optimal offline algorithm. By weak LP duality, the objective function of the Dual LP yields a lower bound on $\OPT$ for any feasible dual solution. We now construct an explicit dual solution $(\alpha,\beta,\gamma)$ whose objective value is roughly equal to the total number of clean and pseudo-clean pages seen by the algorithm. See Section~\ref{subsec:waterfilling} to recall relevant notation and terminology.  
The dual solution is updated at the end of each (global) phase in two stages. Updates in the first stage, denoted $\update(r_0,1)$, are simple and account for stale pages belonging to \full\ agents that got evicted in the most recent local phase for that agent. Updates in the second stage, denoted $\update(r_0,2)$, are more involved and account for the pseudo-clean pages of agents who lost their \fettered\ status in the current phase. The dual solution that we maintain will always be approximately feasible up to $O(1)$ factors, so the objective value of this dual solution serves as a lower bound on $\OPT(T)$ within a constant factor. We remark that the assumption that $T$ marks the end of a phase is without loss of generality since it can lead to at most an additive $O(k)$ loss in the lower bound. Formally, we show the following. 


\begin{theorem} \label{thm:dualguarantee}
Let $T$ denote the timepoint when global phase $R_0$ ends, and let  $(R_i)_{i \in [m]}$ denote the corresponding local phase counters. 
Let $(\alpha,\beta,\gamma)$ denote the dual solution that is constructed by the end of time $T$, i.e., the solution that arises from a sequential application of dual updates in the order $\update(1,1),\update(1,2),\update(2,1),\update(2,2),\dots, \update(R_0,1)$, and $\update(R_0,2)$.
We have:
\begin{enumerate}[(a)]
    \item The dual solution is approximately feasible: for any $i$-page $q$ and an integer $a \geq 0$, $\sum_{t \in I(q,a)} (\alpha(t) - \beta(t,i)) - \gamma(q,a) \leq \dualviolation$ holds.
    
    \item The dual objective value of $(\alpha,\beta,\gamma)$ is: 
    \begin{multline*}
    \dualcost(R_0) \triangleq \sum_{t = 1}^{T} (n-k) \alpha(t) - \sum_{t=1}^{T} \sum_{i \in [m]} (n_i - k_i)  \beta(t,i) - \sum_{q \in \U} \sum_{a = 1}^{a(q,T)} \gamma(q,a) \\
    = \sum_{i \in [m]} \sum_{r_i=1}^{R_i} |P(i,r_i-1) \setminus P(i,r_i)| + \sum_{r_0 = 1}^{R_0} \sum_{i \in \I(r_0-1) \setminus \I(r_0)} (|P(i,r_i-1) \cup P(i,r_i)| - k_i).
    \end{multline*}
\end{enumerate}
\end{theorem}

We first show how Theorem~\ref{thm:dualguarantee} implies that our fractional algorithm $\A$ is $O(\log k)$-competitive. 

\begin{proof}[Proof of Theorem~\ref{thm:fraconline}]
In Theorem~\ref{thm:potential-function} we proved the following upper bound on the cost incurred by $\A$ for processing the first $T$ page requests:
\[
\cost_{\A}(\sigma) \leq 2 \ln(1+k) \cdot \Bigl(mk + \sum_{t=1}^T \bone_{{p_t} \text{ is pseudo-clean}} + \sum_{i \in [m]} \sum_{r_i=1}^{R_i} |P(i,r_i-1) \setminus P(i,r_i)|\Bigr).
\]
Clearly, the second nontrivial term in the above cost-expression matches the first term in the expression for $\dualcost(R_0)$.
Now consider an arbitrary global phase $r_0 \in \{1,\dots,R_0\}$ and a timestep $t$ in this phase. By definition, a pseudo-clean page $p_t$ is necessarily stale, i.e., $p_t \in P(i_t,r_{i_t}-1)$ holds, and it must be that agent $i_t$ was \fettered\ at the start of phase $r_0$ but is \full\ by time $t$. Therefore, $i_t \in \I(r_0-1) \setminus \I(r_0)$ and the following holds: 
\[
|P(i,r_i-1) \cup P(i,r_i)| - k_i \geq |P(i,r_i)| - k_i \geq | \{ t \in \text{phase } r_0 : p_t \text{ is pseudo-clean} \} |.
\]
In the above, the final inequality is because among all pages in $P(i,r_i)$ (w.r.t. the order in which they were marked by $\A$), the first $k_i$ pages are not pseudo-clean.
Thus, the first nontrivial term in the cost-expression for $\A$ can be bounded by the second term in $\dualcost(R_0)$.

Overall, we have shown that $\cost_{\A}(\sigma) \leq 2\ln(1+k) \cdot \bigl( mk + \dualcost(R_0) \bigr)$ holds. Since the dual solution is $O(1)$-feasible, we get that $\A$ is $O(\log k)$-competitive. 
\end{proof}

We now furnish the details of our dual updates. Initially, all our dual variables $\{\alpha(t)\}$, $\{\beta(i,t)\}$, $\{\gamma(q,a)\}$ with $t \in [T], i \in [m], q \in \U, a \in [a(q,T)]$ are set to zero. We assume that the dual updates are applied in the sequence given in Theorem~\ref{thm:dualguarantee}. That is, the set of updates in $\{\update(r_0,s)\}_{r_0 \in [R_0], s \in \{1,2\}}$ are applied in increasing order of $r_0$ and within each phase first stage updates are applied first.
With a slight abuse of notation, let $\dualcost(r_0,s)$ denote the objective value of the dual solution right after updates until $\update(r_0,s)$ (inclusive) have been applied where $r_0 \in [R_0], s \in \{1,2\}$. Note that $\dualcost(R_0) = \dualcost(R_0,2)$. We also define $\dualcost(0,1) = \dualcost(0,2) := 0$. Throughout our updates, we ensure that the dual objective value never decreases, i.e., $0 \leq \dualcost(1,1) \leq \dualcost(1,2) \leq \dots \leq \dualcost(R_0,1) \leq \dualcost(R_0,2)$ holds. 
We remark that $\beta$ variables may decrease and this only happens in the second stage; However, the $\alpha$ and $\gamma$ variables never decrease.

In Section~\ref{subsec:firststagedual}, we describe the first stage of updates and show that the gain in the dual objective corresponds to the first term in Theorem \ref{thm:dualguarantee}(b). In Section~\ref{subsec:secondstagedual}, we describe the second stage of updates and show that the gain in the dual objective corresponds to the second term in Theorem \ref{thm:dualguarantee}(b). Lastly, in Section~\ref{subsec:dualfeasible}, we show that the dual solution that we maintain is always feasible up to constant factors and thus complete the proof of Theorem \ref{thm:dualguarantee}.

\paragraph*{First Stage of Dual Updates} \label{subsec:firststagedual}

Fix a phase $r_0 \in [R_0]$ and consider the set $\I(r_0) \subsetneq [m]$ of agents that are designated as \fettered\ at the end of phase $r_0$. Let $C(r_0)$ denote the set of timesteps $t$ (in this phase) when the following two conditions hold: (a) $y^t_{p_t} = 1$ in the fractional algorithm $\A$ just before $p_t$ is requested; and (b) $i_t$ is not \fettered\ at time $t$. Define $\ell^{(r_0)} \triangleq |C(r_0)|$. It is not hard to see that the following is an equivalent expression for $\ell^{(r_0)}$.
\begin{equation} \label{eq:globalcleanpages}
\ell^{(r_0)} := \sum_{i \notin \I(r_0-1) \cup \I(r_0)} |P(i,r_i) \setminus P(i,r_i-1)| + \sum_{i \in \I(r_0-1) \setminus \I(r_0)} (|P(i,r_i)| - k_i).
\end{equation}
Observe that for agents who are \full\ both at the start and end of phase $r_0$, $\ell^{(r_0)}$ counts all their clean pages. However, for agents who were \fettered\ at the start of this phase but are no longer \fettered\ by the end, $\ell^{(r_0)}$ only counts clean and pseudo-clean pages that are requested \emph{after} the agent has become \full.
Roughly speaking, the motivation for the definition of $\ell^{(r_0)}$ comes from the intuition that an offline algorithm should incur, on an average, a cost of $\Omega(\ell^{(r_0)})$ to serve page requests in phase $r_0$. 


\subparagraph*{Description of \boldmath  $\update(r_0,1)$:}
 For each time $t \in C(r_0)$, we separately apply the following updates. First, we increase $\alpha(t)$ by $1/\ell^{(r_0)}$. Next, we increase $\beta(t,i)$ by $1/\ell^{(r_0)}$ for every agent $i \in \I(r_0)$. Last, for each agent $i \notin \I(r_0)$, we increase $\gamma(q,a(q,t))$ by $1/\ell^{(r_0)}$ for every $i$-page $q \in \U(i) \setminus (P(i,r_i-1) \cup P(i,r_i))$.

It will be clear from the description of our updates that the $\alpha$ and $\beta$ variables that were modified in $\update(r_0,1)$ were previously at $0$. However, no such guarantee holds for the affected $\gamma$ variables. We also remark that the same $\gamma(q,a)$ variable can be increased more than once during $\update(r_0,1)$; this happens when there are multiple times $t \in C(r_0)$ with the same $a(q,t)$ value. In fact, since the $\gamma(q,a)$ variables arise from intervals $I(q,a)$ that can possibly span across multiple phases, it is possible that the same $\gamma$ variable is increased by different $1/\ell^{(r_0)}$ amounts across different $\update(r_0,1)$ steps.

For convenience, let $t \in \text{phase } r_0$ be a shorthand for all timepoints in phase $r_0$. The following result will be useful to us. 

\begin{lemma} \label{lem:totalbetaupdate} 
Let $(\alpha,\beta,\gamma)$ denote the dual solution that is obtained right after $\update(r_0,1)$ has been applied. We have: (a) $\sum_{t \in \text{phase } r_0} \alpha(t) = 1$; and (b) $\sum_{t \in \text{phase } r_0} \beta(t,i) = 1$ for any agent $i \in \I(r_0)$.
\end{lemma}
\begin{proof}
Follows directly from our choice of $\ell^{(r_0)} = |C(r_0)|$. 
\end{proof}

Our key technical result in this section is that the gain in the dual objective value that comes from $\update(r_0,1)$ is equal to the number of stale pages owned by \full\ agents that were not requested in their most recent local phases. For convenience, 
we use the prefix $\Delta$ to refer to changes that occured during $\update(r_0,1)$.


\begin{lemma} \label{lem:firstupdate}
We have $\Delta\dualcost(r_0,1) = \sum_{i \notin \I(r_0)} |P(i,r_i-1) \setminus P(i,r_i)|$, where $\Delta \dualcost(r_0, 1) \triangleq \dualcost(r_0, 1) - \dualcost(r_0-1, 2)$ is the change in the dual objective after $\update(r_0, 1)$.  
\end{lemma}

\begin{proof}
Since the only affected $\alpha(t)$ and $\beta(t,i)$ variables have are those with $t \in C(r_0)$ and they are all increased by exactly $1/|C(r_0)|$, we get:
\begin{equation*} \label{eq:firstupdatepart1}
\begin{aligned} 
\Delta \dualcost(r_0,1) & = \sum_t (n-k) \Delta \alpha(t) - \sum_{t,i} (n_i - k_i) \Delta \beta(t,i) - \sum_{q,a} \Delta \gamma(q,a) \\
& = (n-k)  - \sum_{i \in \I(r_0)} (n_i - k_i)  - \sum_{q,a} \Delta \gamma(q,a) \\ 
& = \Bigl( \sum_{i \notin \I(r_0)} n_i \Bigr) - k + \Bigl( \sum_{i \in \I(r_0)} k_i \Bigr) - \sum_{q,a} \Delta \gamma(q,a) \\ 
\end{aligned}
\end{equation*}

Now observe that for every $t \in C(r_0)$ and $i \notin \I(r_0)$, the $\update(r_0,1)$ step increases the $\gamma(q,a)$ variable corresponding to exactly $n_i - |P(i,r_i-1) \cup P(i,r_i)|$ unique $i$-pages, each by an amount $1/\ell^{(r_0)}$. So we have $\sum_{q,a} \Delta \gamma(q,a) = (1/\ell^{(r_0)}) \cdot \sum_{t \in C(r_0)} \sum_{i \notin \I(r_0)} (n_i - |P(i,r_i-1) \cup P(i,r_i)|) = \sum_{i \notin \I(r_0)} (n_i - |P(i,r_i-1) \cup P(i,r_i)|)$.
Substituting back into the equation above, we get:
\begin{equation*} \label{eq:firstupdatepart2}
\begin{aligned} 
\Delta \dualcost(r_0,1) & = \Bigl( \sum_{i \notin \I(r_0)} n_i \Bigr) - k + \Bigl( \sum_{i \in \I(r_0)} k_i \Bigr) - \sum_{i \notin \I(r_0)} (n_i - |P(i,r_i-1) \cup P(i,r_i)|) \\
& = \Bigl( - k + \sum_{i \in \I(r_0)} k_i + \sum_{i \notin \I(r_0)} |P(i,r_i)| \Bigr) + \sum_{i \notin \I(r_0)} |P(i,r_i-1) \setminus P(i,r_i)|
\end{aligned}
\end{equation*}
The lemma follows from observing that the above group of terms within the parentheses is $0$: this is because the cache (of size $k$) at the end of phase $r_0$ consists exactly $k_i$ (fractional) pages for \fettered\ agents  $i \in \I(r_0)$ and exactly $|P(i,r_i)|$ (integral) pages for \full\ agents $i \notin \I(r_0)$.
\end{proof}




\paragraph*{Second Stage of Dual Updates} \label{subsec:secondstagedual}

We now describe the second stage of dual updates that are carried out at the end of each phase $r_0 \in [R_0]$. Unlike the first stage, where we only increased the $\alpha$ and $\beta$ variables of time steps in phase $r_0$, in the second stage we decrease the $\beta$ variables of time steps in the previous phase $r_0-1$.

\subparagraph*{Description of \boldmath  $\update(r_0,2)$:} These dual updates correspond to agents that were \fettered\ at the end of phase $r_0-1$ but are no longer \fettered\ at the end of phase $r_0$. Consider an agent $i \in \I(r_0-1) \setminus \I(r_0)$. For every time $t$ in phase $r_0-1$ with $\beta(t,i) > 0$ (i.e., $t \in C(r_0-1)$), we do the following: we increase $\gamma(q,a(q,t))$ by $\beta(t,i)$ for all $i$-pages $q \in \U(i) \setminus \bigl(P(i,r_i-1) \cup P(i,r_i)\bigr)$ followed by resetting $\beta(t,i)$ to $0$. 

For clarity, we note the following: (i) resetting $\beta(t,i)$ to zero is the only dual update when a variable is \emph{decreased}; (ii) the $\beta(t,i)$ updates are applied to timepoints in phase $r_0-1$ (i.e., the previous phase); and (iii) the $\gamma(q,a(q,t))$ variables that we updated above were unchanged while applying  $\update(r_0-1,1)$ at the end of phase $r_0-1$ because their owner $i$ was designated as \fettered\ at that time. 
The reason for decreasing the $\beta(t,i)$ variables is that it leads to an increase in the dual objective, which will be needed to pay for costs associated with pseudo-clean pages.
We formalize this in the following lemma.




\begin{lemma}
\label{lem:dual-secondstage}
We have $\Delta \dualcost(r_0, 2) =  \sum_{i \in \I(r_0-1) \setminus \I(r_0)} \left(|P(i,r_i-1) \cup P(i,r_i)| - k_i\right)$ where $\Delta \dualcost(r_0, 2) \triangleq \dualcost(r_0, 2) - \dualcost(r_0, 1)$ is the change in the dual objective after $\update(r_0, 2)$.
\end{lemma}
\begin{proof}
Fix an agent $i \in \I(r_0-1) \setminus \I(r_0)$.
In Lemma~\ref{lem:totalbetaupdate} we showed that after $\update(r_0-1, 1)$, $\sum_{t \in \text{phase } r_0-1} \beta(t,i)$ = 1 and $\beta(t,i) \in \{0, 1/\ell^{(r_0-1)}\}$. Consider any time $t$ in phase $r_0-1$ with $\beta(t, i) > 0$. By the definition of $\update(r_0, 2)$, we decrease $\beta(t, i)$ by $1/\ell^{(r_0-1)}$ while increasing $\gamma(q, a(q, t))$ by the same amount for all pages $q \in \U(i) \setminus \bigl( P(i, r_i-1) \cup P(i, r_i) \bigr)$. Recalling the coefficients in the dual objective function, we see that the updates corresponding to agent $i$ increases the dual objective by exactly:
\[(n_i - k_i) - |\U(i) \setminus (P(i,r_i) \cup P(i,r_i -1))| = |(P(i,r_i) \cup P(i,r_i -1))| - k_i. \qedhere\]
\end{proof}


\paragraph*{Approximate Dual Feasibility} \label{subsec:dualfeasible}

We finish this section by showing that the dual solution is always approximately feasible.


\begin{lemma} \label{lem:dualfeasible}
Let $(\alpha,\beta,\gamma)$ denote the dual solution that is obtained right after $\update(r_0,s)$ has been applied for some $r_0 \in [R_0]$ and $s \in \{0,1\}$. 
For any $i$-page $q$ and an integer $a \geq 1$ satisfying $a \leq a(q,T)$, we have $\sum_{t \in I(q,a)} (\alpha(t) - \beta(t,i)) - \gamma(q,a) \leq \dualviolation$.
\end{lemma}

\begin{proof}
First of all, for the purposes of this proof, the specific values of $r_0$ and $s$ are irrelevant, so we ignore them. Fix some $i$-page $q$ and an integer $a$ satisfying $a \leq a(q,T)$. Recall that $I(q,a) = \{t_{q,a}+1, \dots, t_{q,a+1}-1 \}$, where $t_{q,a'}$ denotes the time when $q$ is requested for the $a'$th time; We redefine $t_{q,a+1}$ to be $T+1$ if $t_{q,a+1} > T$ holds. 

The lemma holds trivially if $I(q,a)$ is empty, so we assume otherwise. Let $r^b_0,r^e_0 \in [R_0]$ denote the global phases that contain timesteps $t_{q,a}+1$ and $t_{q,a+1}-1$, respectively. Clearly, $r^b_0 \leq r^e_0$. Another easy case of the lemma is when $r^e_0 \leq r^b_0 + 1$ holds. The desired conclusion follows easily because all the dual variables are nonnegative and the sum of all $\alpha(t)$ variables in any phase is at most $1$ (by Lemma~\ref{lem:totalbetaupdate}). Formally, 
\[
\sum_{t \in I(q,a)} \bigl(\alpha(t) - \beta(t,i) \bigr) - \gamma(q,a) \leq \sum_{t \in \text{phase } r^b_0} \alpha(t) + \sum_{t \in \text{phase } r^e_0} \alpha(t) \leq 2.
\]

Now suppose that $r^b_0 + 2 \leq r^e_0$ holds. Define $Z := \{r^b_0+1,\dots,r^e_0-1\}$. 
Repeating the above calculation, we get: 
\[
\sum_{t \in I(q,a)} (\alpha(t) - \beta(t,i)) - \gamma(q,a) \leq 2 + \left\{ \sum_{r_0 \in Z} \sum_{t \in \text{phase } r_0} \bigl(\alpha(t) - \beta(t,i) \bigr) \right\} - \gamma(q,a),
\] 
so the crux of the lemma is to show that the sum of $\delta(t) \triangleq \alpha(t) - \beta(t,i)$ over timepoints spanning phases in $Z$ is not much larger than $\gamma(q,a)$. For a phase $r_0 \in Z$, we overload the notation $\delta(r_0)$ to mean $\sum_{t \in \text{phase } r_0} \delta(t)$. Note that by nonnegativity of $\beta$ variables and Lemma~\ref{lem:firstupdate}, $\delta(r_0) \leq 1$ for every $r_0 \in Z$. We do a case analysis on phase $r_0 \in Z$ to get a better handle on the changes that happens during our dual update procedures.

\begin{enumerate}[(a)]
    \item Suppose that $i \in \I(r_0) \cap \I(r_0+1)$ holds. Since $i$ is \fettered\ by the end of phase $r_0$, we know that any increase in $\alpha(t)$ (as part of $\update(r_0,1)$) for some $t \in C(r_0)$ is accompanied with the same increase in $\beta(t,i)$. Since $i \in \I(r_0+1)$ holds, $\update(r_0+1,2)$ does not decrease/reset any of the $\{\beta(t,i)\}_{t \in \text{phase } r_0}$ variables to $0$. Thus, $\delta(r_0) = 0$ holds. Note that there can be an arbitrary number of phases $r_0$ that fall under this case, but this is not a problem for us since $\delta(r_0) = 0$. 
    
    \item Suppose that $i \in \I(r_0) \setminus \I(r_0+1)$ and $q \in P(i,r_i-1) \cup P(i,r_i)$ hold. We rely on the trivial bound $\delta(r_0) \leq 1$ for this case. Since $i$ is \fettered\ by the end of phase $r_0$ but is \full\ by the end of phase $r_0+1$, the local phase counter $r_i$ goes up by $1$ at the end of phase $r_0+1$, and subsequently the $P(i,\cdot)$ set gets updated with some new  collection of marked $i$-pages. By definition of $I(q,a)$, there are no page-requests for $q$ during any of the phases in $Z$. Thus, there can be at most $2$ local phase increments for agent $i$ before $q$ gets dropped from the $P(i,\cdot)$ set; By the design of $\A$, page $q$ cannot enter any of the future $P(i,r_i)$ until the next time it is requested, which does not happen during any of the phases in $Z$.
    
    \item Suppose that $i \in \I(r_0) \setminus \I(r_0+1)$ and $q \notin P(i,r_i-1) \cup P(i,r_i)$ hold. Similar to case~(a) above, we know that any increase in $\alpha(t)$ (as part of $\update(r_0,1)$) for some $t \in C(r_0)$ is accompanied with the same increase in $\beta(t,i)$. Now, although $\update(r_0+1,2)$ decreases/resets all the $\{\beta(t,i)\}_{t \in C(r_0)}$ variables to $0$, it also increases $\gamma(q,a)$ by the same amount since $q \notin P(i,r_i-1) \cup P(i,r_i)$. Thus, $\sum_{t \in \text{phase } r_0} \alpha(t)$ equals the increase in $\gamma(q,a)$ due to $\update(r_0+1,2)$. So, the difference is essentially $0$.
    
    \item Suppose $i \notin \I(r_0)$ and $q \in P(i,r_i-1) \cup P(i,r_i)$ hold. We rely on the trivial bound $\delta(r_0) \leq 1$ for this case. Since $i$ is \full\ by the end of phase $r_0$, the local phase counter $r_i$ goes up by $1$ at the end of phase $r_0$. Repeating the argument from case~(c), there can be at most $1$ more local phase increment for agent $i$ before $q$ gets dropped from the $P(i,\cdot)$ set for the rest of the phases in $Z$.
    
    \item Suppose $i \notin \I(r_0)$ and $q \notin P(i,r_i-1) \cup P(i,r_i)$ hold. Since $i$ is \full\ by the end of phase $r_0$ and $q$ is not in $P(i,r_i-1) \cup P(i,r_i)$, we know that any increase in $\alpha(t)$ (as part of $\update(r_0,1)$) for some $t \in C(r_0)$ is accompanied with the same increase in $\gamma(q,a)$.
    Thus, $\sum_{t \in \text{phase } r_0} \alpha(t)$ equals the increase in $\gamma(q,a)$ due to $\update(r_0,1)$. So, the difference is essentially $0$.
    
\end{enumerate}

From the above case analysis, it follows that $\{\sum_{r_0 \in Z} \sum_{t \in \text{phase } r_0} \bigl(\alpha(t) - \beta(t,i) \bigr)\} - \gamma(q,a)$ is bounded by the number of phases $r_0 \in Z$ for which case~(b)~or~(d) hold. Since we argued that there can be at most $3$ such occurences, $\sum_{t \in I(q,a)} (\alpha(t) - \beta(t,i)) - \gamma(q,a) \leq 2 + 3 =  \dualviolation$ holds.
\end{proof}

We now prove the main theorem in this section by combining the above lemmas. 

\begin{proof}[Proof of Theorem \ref{thm:dualguarantee}] The first part of the theorem follows from Lemma \ref{lem:dualfeasible}. 
The second part follows directly from Lemmas \ref{lem:firstupdate} and \ref{lem:dual-secondstage} since we have
$\dualcost(R_0) = \sum_{r_0=1}^{R_0} \bigl( \Delta \dualcost(r_0, 1) + \Delta \dualcost(r_0, 2) \bigr)$
\end{proof}
\section{Rounding} \label{sec:rounding} 

In this section we show how to convert the fractional Algorithm \ref{alg:water-filling} into a randomized integral online algorithm for \emph{Caching with Reserves}, each step of which runs in polynomial time. 

The algorithm maintains a uniform distribution on $N = k^3$ valid cache states. In each step, these states are updated based on the actions of the fractional algorithm. The randomized algorithm selects one of these states uniformly at random in the beginning, and then follows it throughout the run.

Initially, all $N$ cache states in the distribution are the same as the initial cache state in Algorithm \ref{alg:water-filling}. 
Given the fractional algorithm values $x^t_p$ after each page request, the distribution is updated in two steps. First, we produce a discretized version of these fractions, $\tilde{x}^t_p$, which are also a feasible fractional solution. This is based on the technique in \cite{adamaszek2018log}. Second, we update the $N$ cache states so that all of them remain valid, and for each page $p$, exactly $N \cdot \tilde{x}^t_p$ of the states contain $p$. This is based on the technique in \cite{ibrahimpur2022caching}. We note that the rounding procedure can be done online, as it does not need the knowledge of any future page requests.

\subsection{Discretization Procedure}

\newcommand{\round}[2]{\left \lfloor #1 \right \rfloor_{#2}}

In this subsection, we explain how to perform the first step: discretizing the fractional algorithm's values $x^t_p$ into $\tilde{x}^t_p$ which are multiples of $\frac{1}{N}$. Our procedure is quite simple: we iterate over the pages in any order $\pi$ that arranges all pages belonging to the same agent consecutively (i.e., order the agents arbitrarily and order each agent's pages arbitrarily, but do not interleave pages from different agents). Then for $i \in [|\U|]$, set:
\begin{align*}
  \tilde{x}^t_{\pi(i)} &\triangleq
    \round{\sum_{j=1}^i x^t_{\pi(j)}}{1/N} - \round{\sum_{j=1}^{i-1} x^t_{\pi(j)}}{1/N}
\end{align*}
where $\round{a}{b}$ denotes rounding $a$ down to the nearest multiple of $b$; formally: $\round{a}{b} \triangleq b \left\lfloor a/b \right \rfloor$.

\begin{lemma} \label{lem:discretization}
Discretization satisfies the following guarantees:
\begin{enumerate}
  \item $\tilde{x}^t_p$ is a multiple of $1/N$
  \item $\left| \tilde{x}^t_p - x^t_p \right| < 1/N$
  \item for each agent $i$, $\left| \sum_{p \in \U(i)} \tilde{x}^t_p - \sum_{p \in \U(i)} x^t_p \right| < 1/N$
  \item if $x^t_p \in \{0, 1\}$, then $\tilde{x}^t_p = x^t_p$
\end{enumerate}
\end{lemma}

\begin{proof}
The first guarantee, $\tilde{x}^t_p$ is a multiple of $1/N$, follows from $\round{\cdot}{1/N}$ being a multiple of $1/N$ by definition.

We now prove the second guarantee, $\left| \tilde{x}^t_p - x^t_p \right| < 1/N$. Let $i$ be $\pi^{-1}(p)$, and we get:
\begin{align*}
  a - \round{a}{b}
    &= a - b \left\lfloor a/b \right \rfloor \\
    &\in \bigl[ a - b(a/b), a - b(a/b-1) \bigr) \\
    &= [ 0, b ) \\
  \tilde{x}^t_p - x^t_p
    &= \tilde{x}^t_{\pi(i)} - x^t_{\pi(i)} \\
    &= \underbrace{\left[
        \round{\sum_{j=1}^i x^t_{\pi(j)}}{1/N} -
        \round{\sum_{j=1}^{i-1} x^t_{\pi(j)}}{1/N}
      \right]}_{\tilde{x}^t_{\pi(i)}} - \underbrace{\left[
        \sum_{j=1}^i x^t_{\pi(j)} - \sum_{j=1}^{i-1} x^t_{\pi(j)}
      \right]}_{x^t_{\pi(i)}} \\
    &= \left[ 
        \sum_{j=1}^{i-1} x^t_{\pi(j)} -
        \round{\sum_{j=1}^{i-1} x^t_{\pi(j)}}{1/N}
      \right] - \left[
        \sum_{j=1}^i x^t_{\pi(j)} -
        \round{\sum_{j=1}^i x^t_{\pi(j)}}{1/N}
      \right] \\
    &\in \left(-\frac{1}{N}, \frac{1}{N}\right) \\
  \left| \tilde{x}^t_p - x^t_p \right| &< \frac{1}{N}
\end{align*}

The proof of the third guarantee is essentially the same. It is here we use the fact that each agent's pages are consecutive in $\pi$. Suppose our agent's pages are $\{\pi(\ell), ..., \pi(h)\}$:
\begin{align*}
  \sum_{p \in \U(i)} \tilde{x}^t_p - \sum_{p \in \U(i)} x^t_p
    &= \sum_{i=\ell}^h \tilde{x}^t_{\pi(i)} - \sum_{i=\ell}^h x^t_{\pi(i)} \\
    &= \left[
        \round{\sum_{j=1}^h x^t_{\pi(j)}}{1/N} -
        \round{\sum_{j=1}^{\ell-1} x^t_{\pi(j)}}{1/N}
      \right] - \left[
        \sum_{j=1}^h x^t_{\pi(j)} - \sum_{j=1}^{\ell-1} x^t_{\pi(j)}
      \right] \\
    &= \left[
        \sum_{j=1}^{\ell-1} x^t_{\pi(j)} -
        \round{\sum_{j=1}^{\ell-1} x^t_{\pi(j)}}{1/N}
      \right] - \left[
        \sum_{j=1}^h x^t_{\pi(j)} -
        \round{\sum_{j=1}^h x^t_{\pi(j)}}{1/N}
      \right] \\
    &\in \left(-\frac{1}{N}, \frac{1}{N}\right) \\
  \left| \sum_{p \in \U(i)} \tilde{x}^t_p - \sum_{p \in \U(i)} x^t_p \right| &< \frac{1}{N}
\end{align*}
The fourth guarantee follows by noting that if $x^t_{\pi(i)}$ is an integer, then the $\round{\cdot}{1/N}$ operator decreases $\sum_{j=1}^i x^t_{\pi(j)}$ as much as it decreases $\sum_{j=1}^{i-1} x^t_{\pi(j)}$.
\end{proof}



\begin{corollary} \label{cor:reserve}
  If $\{x^t_p\}$ satisfy total cache capacity ($\sum_{p \in \U} x^t_p \le k$) and reserve requirements ($\sum_{p \in \U(i)} x^t_p \ge k_i$), then so do  $\{\tilde{x}^t_p\}$.
\end{corollary}

\begin{proof}
  The total cache capacity constraint continues to hold due to a telescoping argument:
  \begin{align*}
    \sum_{p \in \U} \tilde{x}^t_p
      &= \sum_{i \in [|\U|]} \tilde{x}^t_{\pi(i)} 
      = \sum_{i \in [|\U|]} \left[
        \round{\sum_{j=1}^i x^t_{\pi(j)}}{1/N} - \round{\sum_{j=1}^{i-1} x^t_{\pi(j)}}{1/N}
      \right] \\
      &= \round{\sum_{j=1}^{|\U|} x^t_{\pi(j)}}{1/N} 
      \le \sum_{j=1}^{|\U|} x^t_{\pi(j)}
      \le k
  \end{align*}

  Next, we will prove that reserve cache sizes are satisfied. For the sake of contradiction, suppose that for some agent $i$, $\sum_{p \in \U(i)} \tilde{x}^t_p < k_i$. Since the right-hand side of this inequality is an integer and therefore a multiple of $1/N$, the left-hand side, which is also a multiple of $1/N$ due to being a sum of multiples of $1/N$ (by Lemma \ref{lem:discretization}'s first guarantee), must be at least a full multiple of $1/N$ less than the right-hand side: $\sum_{p \in \U(i)} \tilde{x}^t_p \le k_i - 1/N$. But this contradicts Lemma \ref{lem:discretization}'s third guarantee, $\left| \sum_{p \in \U(i)} \tilde{x}^t_p - \sum_{p \in \U(i)} x^t_p \right| < 1/N$. Therefore for all agents $i$, $\sum_{p \in \U(i)} \tilde{x}^t_p \ge k_i$, completing the proof.
\end{proof}

\begin{corollary} \label{cor:cost}
  $\sum_p \left| \tilde{x}^t_p - x^t_p \right| \le k^2 / N$
\end{corollary}

\begin{proof}
  Without loss of generality, there are at most $k$ agents since we can combine all agents that do not have any reserve. Each agent $i$ has at most $k$ fractional pages by the algorithm (the $i$-pages that were in cache when $i$'s local phase began). The $x$ value of each fractional page is distorted by at most $1/N$ by Lemma \ref{lem:discretization}'s second guarantee, while not being distorted for non-fractional pages by the fourth guarantee. This completes the proof.
\end{proof}


\begin{lemma} \label{lem:cost-discret}
  Let $x^t_p$ and $x^{t+1}_p$ be the amounts of each page $p$ in cache in the fractional algorithm for two consecutive time steps, and $\tilde{x}^t_p$ and $\tilde{x}^{t+1}_p$ be the corresponding discretized values. Then the cost of cache update from $\tilde{x}^t$ to $\tilde{x}^{t+1}$ (call it $\tilde{c}$) is at most twice the cost of cache update from $x^t$ to $x^{t+1}$ (call it $c$).
\end{lemma}

\begin{proof}
  Lemma \ref{lem:ylarge} implies that either $c=0$ (i.e., the requested page was already in cache and there is no change to the cache state), or $c \ge 1/k$. In the first case, there is no change to the discretized cache state either, so $\tilde{c}=0$. So we focus on the second case. By triangle inequality, for any page $p$,
  \[
      |\tilde{x}^t_p - \tilde{x}^{t+1}_p| \le 
      |\tilde{x}^t_p - x^t_p| + |x^t_p - x^{t+1}_p| + |x^{t+1}_p - \tilde{x}^{t+1}_p|.
  \]
  We note that since cost is incurred for adding pages to cache, and both the original fractional solution and the discretized one add as much page mass to cache as they evict, $2c =  \sum_p |x^t_p - x^{t+1}_p|$, and similarly for $\tilde{c}$. Summing the above inequality over $p$, we get
  \[
    2\tilde{c} =  \sum_p |\tilde{x}^t_p - \tilde{x}^{t+1}_p| \le
    \sum_p |x^t_p - x^{t+1}_p| + 2k^2/N = 2 c + 2/k \le 4c,
  \]
  where we used $\sum_p (|\tilde{x}^t_p - x^t_p|+ |x^{t+1}_p - \tilde{x}^{t+1}_p|) \leq 2k^2/N$ by Corollary \ref{cor:cost}, then $N = k^3$, then $c \ge 1/k$.
\end{proof}

\subsection{Updating the Distribution of Cache States}

In this subsection, we explain how to perform the second step: updating the $N$ cache states. We would like (i) all cache states to be valid, (ii) exactly $N \cdot \tilde{x}^t_p$ (integral due to our discretization step) of the states to contain page $p$, and (iii) to not use too many evictions.

Formally, let $\X$ be a set of $N$ cache states with $k$ pages each, which corresponds to the discretized values $\tilde{x}^t_p$ for time step $t$. Given the discretized values $\tilde{x}^{t+1}_p$ for time step $t+1$, we show how to transform $\X$ into $\X'$, in which each page $p$ appears in exactly $N \cdot \tilde{x}^{t+1}_p$ cache states and such that each cache state satisfies all the reserve requirements.

Let $P$ be a multiset of pages whose fraction in the cache increased from time $t$ to $t+1$, with each page $p$ appearing $\max(0, (\tilde{x}^{t+1}_p - \tilde{x}^{t}_p)N)$ times. Let $Q$ be an analagous multiset for decreases, with each page appearing $\max(0, (\tilde{x}^{t}_p - \tilde{x}^{t+1}_p)N)$ times. Since the total amount of pages in the cache is unchanged, $|P| = |Q|$. We find a matching between pages in $P$ and $Q$ and use it to transform $\X$ into $\X'$ gradually, one pair at a time. The matching is constructed as follows. First, any pages from $P$ and $Q$ that belong to the same agent are matched up. Then, the remaining pages in $P$ and $Q$ are matched up arbitrarily.

\begin{lemma} \label{lem:req}
Let $(p_1, q_1), (p_2, q_2), ...$ be the matching between $P$ and $Q$ described above. Then for any $j$, the fractional solution that adds $1/N$ fraction of pages $p_1,...,p_j$ to $\tilde{x}^{t}$ and removes $1/N$ fraction of pages $q_1,...,q_j$ from it satisfies all the reserve requirements.
\end{lemma}


\begin{proof}
As the pairs are applied, the fractional solution transforms from $\tilde{x}^{t}$ to $\tilde{x}^{t+1}$, both of which satisfy the reserve requirements (by Corollary \ref{cor:reserve}). Inductively, assume that when pairs up to $(p_{j-1}, q_{j-1})$ have been applied, the reserved requirements are still satisfied. If $p_j$ and $q_j$ belong to the same agent $i$, then amount of $i$-pages in the cache doesn't change, thus still satisfying the requirement. If $p_j$ and $q_j$ belong to different agents, then we must be in the second phase of our matching construction. This means that there are no more pages in $P$ belonging to $ag(q_j)$, so the amount of $ag(q_j)$-pages in cache does not increase while applying the remainder of the matching pairs. This means that after applying $(p_j, q_j)$, the amount of $ag(q_j)$-pages in the cache is somewhere between what it has in $\tilde{x}^{t}$ and in $\tilde{x}^{t+1}$. Since both of those satisfy the reserve requirements, so does the intermediate solution.
\end{proof}

We now show how to modify $\X$ with the next pair $(p, q)$ from the matching. This follows the procedure in \cite{ibrahimpur2022caching}, with the difference that we work on a limited number of $N$ sets, and the amount of increase in $p$ and decrease in $q$ is fixed at $1/N$.

Let $\X$ be the current set of cache states (possibly modified by the previous page pairs). 
If there is a cache state $S\in\X$ such that $p\notin S$ and $q\in S$, add $p$ to $S$ and remove $q$ from $S$. Otherwise, find cache states $S\in \X$ and $T\in \X$ with $p\notin S$ and $q\in T$,  add $p$ to $S$ and remove $q$ from $T$. Next, move some page $r \in S \setminus T$ from $S$ to $T$ to adjust the set sizes back to $k$.

At this point, each page is in the correct number of cache states. However, reserve requirements could be violated by one page for $ag(q)$ or $ag(r)$ in the cache states from which the corresponding pages were removed.
In such a case, suppose the requirement is violated for agent $i$ in a cache state $V\in\X$. Since, by Lemma \ref{lem:req}, each reserve requirement is satisfied on average, there must be another set $W \in \X$ which has strictly more than $k_i$ pages belonging to agent $i$. 
We move one such page from $W$ to $V$. Now $V$ has $k+1$ pages, so there must be an agent $j$ which has more than $k_j$ pages in $V$. We move one of $j$'s pages from $V$ to $W$ to restore the sizes. 
This completes the update, resulting in new valid sets corresponding to the fractions $\tilde{x}^{t+1}$.

We conclude by bounding the cost of update to $\X$, and thus the expected cost of the randomized algorithm, relative to the cost of the fractional algorithm.

\begin{lemma}\label{lem:cost-rand}
The cost of update to $\X$ is at most 6 times the cost of fractional cache update from $\tilde{x}^t$ to $\tilde{x}^{t+1}$.
\end{lemma}
\begin{proof}
Each pair $(p, q)$ in the matching corresponds to a cost of $1/N$ incurred by the discretized fractional solution. In the updates to sets in $\X$, each time a page is removed from one of the cache states incurs a cost of $1/N$ to the randomized algorithm. At most, the following six removals are done: remove $q$ from $T$; remove $r$ from $S$; two pages each are swapped to fix the reserve requirements for $ag(q)$ and $ag(r)$. 
\end{proof}

The proof of Theorem \ref{thm:randomized} follows by combining Lemmas \ref{lem:cost-discret} and \ref{lem:cost-rand}.

\bibliography{caching}

\appendix
\section{Analysis of Randomized Marking for Unweighted Paging}
\label{app:regular-paging}

In this section, we consider the classical unweighted paging problem and provide a new proof that the randomized marking~\cite{fiat1991competitive} algorithm is $O(\log k)$-competitive.

Recall that the randomized marking algorithm works in phases as follows. Initially all pages in the cache are \emph{unmarked}. When a page $p_t$ is requested at time $t$, if $p_t$ is already in the cache then it is \emph{marked} and no other change occurs. If $p_t$ is not in the cache, then a page $q$ is chosen uniformly at random from the set of \emph{unmarked} pages currently in the cache and evicted. Page $p_t$ is then brought into the cache and \emph{marked}. Once there are no more unmarked pages left for eviction, the current \emph{phase} ends --- all pages in the cache are \emph{unmarked} and a new phase begins.

Let $\RM$ denote the above randomized marking algorithm and let $\cost_{\RM}(\sigma)$ be the total cost incurred on an instance $\sigma$. For analysis, we consider the following fractional version of the marking algorithm:  The algorithm maintains a set of \emph{marked} and \emph{unmarked} pages as before. At any time $t$, let $x_p^t$ denote the fraction of page $p$ in the cache and $y_p^t \triangleq 1 - x_p^t$ denote the fraction of page $p$ outside the cache. When a page $p_t$ is requested, the fractional marking algorithm sets $x_{p_t} = 1$ and uniformly increases $y_q$ for all unmarked pages $q$ in the cache. As in the randomized marking algorithm, a phase ends when all unmarked pages have been fully evicted from the cache. Observe that for any page $p$ and time $t$, the value $y_p^t$ in the fractional marking algorithm exactly equals the marginal probability that page $p$ has been evicted from the cache by the randomized marking algorithm. Let \FM\ denote the fractional marking algorithm, then we have:
$\E[\cost_{\RM}(\sigma)] = \cost_{\FM}(\sigma)$. In the rest of this section, we focus on bounding the total cost incurred by the fractional marking algorithm.

\paragraph*{Potential Function}
Consider any phase $r$ of the fractional marking algorithm and let $P^{r-1}$ denote the set of exactly $k$ pages in the cache at the beginning of the phase and $P^r$ denote the set of $k$ pages in cache at the end of the phase. By definition, $P^r$ is exactly the set of $k$ distinct pages that are requested during phase $r$. Following standard terminology, we say a page $p \in P^r$ is stale if $p \in P^{r-1}$ and call $p$ as clean otherwise. Let $\ell^{(r)} = |P^r \setminus P^{r-1}|$ denote the total number of clean pages in phase $r$. At any time $t$, let $U(t)$ be the set of unmarked pages at time $t$. Consider the following potential function: 

\[\Psi(t) = \sum_{p \in U(t)} \phi(y_p^t)\]
  where  $\phi(h) = 2h\ln(1+kh)$ as earlier. At any time step $t$, let $\Delta \Psi(t)$ be the total change in potential while processing the requested page $p_t$\footnote{In particular, $\Delta \Psi(t)$ does \emph{not} include any change in potential due to a phase change.}. Let $\Delta \Psi(r)$ denote the change in potential due to end of a phase $r$.

\begin{lemma}
At any time $t$, if $y_{p_t}^t < 1$, then $\Delta \Psi(t) \geq y_{p_t}^t$.
\end{lemma}

\begin{proof}
At any time $t$, all unmarked pages in the cache have been evicted the same amount fractionally, i.e. $y_p^t = y_q^t = h^*, \forall p \neq q \in U(t)$. While serving the requested page $p_t$, the potential decreases as page $p_t$ stops contributing to the potential while the potential increases due to the eviction of all other unmarked pages. Page $p_t$ contributed exactly $\phi(h^*) = 2h^*\ln(1+kh^*)$ to the potential at the beginning of this step. On the other hand, to fetch page $p_t$, the other unmarked pages are evicted by a total amount of $h^*$. By Lemma \ref{lem:largeslope}, the potential increases at a rate of at least $(1 + 2\ln(1+kh^*))$ throughout this eviction process. Summing up both the contributions, we get:
\[
    \Delta \Psi(t) \geq -2h^*\log(1+kh^*) + h^*(1 + 2\ln(1+kh^*)) = y_{p_t}^t. \qedhere
\]
\end{proof}

\begin{lemma}
At the end of any phase $r$, the potential drops by $\Delta \Psi(r) = -2\ell^{(r)} \ln(1+k)$.
\end{lemma}

\begin{proof}
At the end of phase $r$, all pages in $P^{r-1} \setminus P^{r}$ are unmarked and have been fully evicted from the cache. Each such page contributed exactly $\phi(1) = 2\ln(1+k)$ to the potential just before the phase ended. Once phase $r$ ends (and phase $r+1$ begins), all these pages stop contributing to the potential. The claim follows by noting that $|P^{r-1} \setminus P^{r}| = |P^{r} \setminus P^{r-1}| = \ell^{(r)}$.
\end{proof}

The previous two lemmas allow us to bound the total cost incurred by the fractional marking algorithm purely in terms of the total number of clean pages.
\begin{lemma}
We have $\cost_{\FM}(\sigma) \leq 2k\ln(1+k) + \left(1+2\ln(1+k)\right) \sum_{r=1}^R \ell^{(r)}$.
\end{lemma}
\begin{proof}
We have the following chain of equations and inequalities:
\begin{align*}
    \cost_{\FM}(\sigma) &= \sum_{t \in [T]} y_{p_t}^t = \sum_{t \in [T] : y_{p_t}^t < 1} y_{p_t}^t + |\{t \in [T] : y_{p_t}^1 = 1\}|\\
    &= \sum_{t \in [T] : y_{p_t}^t < 1} y_{p_t}^t + \sum_{r=1}^R \ell^{(r)}\\
    &\leq \sum_{t \in [T] : y_{p_t}^t < 1} \Delta \Psi(t) + \sum_{r=1}^R \ell^{(r)} \\
    &= \Psi(T) - \Psi(0) - \sum_{r=1}^R \Delta \Psi(r)+ \sum_{r=1}^R \ell^{(r)} \\
    &\leq 2k\ln(1+k) + \left(1+2\ln(1+k)\right) \sum_{r=1}^R \ell^{(r)}.
\end{align*}
\end{proof}

\begin{theorem}
The randomized marking algorithm is $O(\log k)$-competitive.
\end{theorem}
\begin{proof}
On any instance $\sigma$, we have $\E[\cost_{\RM}(\sigma)] = \cost_{\FM}(\sigma)$. The desired bound on the competitive ratio follows by noting that $\cost_{OPT}(\sigma) \geq (1/2) \cdot \sum_{r=1}^R \ell^{(r)}$ as shown by \cite{fiat1991competitive}. We can also use the dual-fitting argument presented in Section \ref{sec:dual-fitting} to show that even the fractional optimal solution is at least $(1/2) \cdot \sum_{r=1}^R \ell^{(r)}$ but we skip the details for brevity.
\end{proof}
 
\end{document}